\documentclass[a4paper,11pt]{article}
\usepackage{jheppub} 
\usepackage{lineno}
\usepackage{moresize}
\usepackage{wasysym}
\usepackage{mathtools}
\usepackage{twemojis}
\usepackage{amsthm}
\usepackage{tipa}

\input{macros.sty}
\input{draw.sty}
\newcommand{\invrho}[1]{\underline{#1}}

\title{\boldmath Torus algebra and logical operators at low energy}

\author{Ying Chan, Tian Lan, and Linqian Wu}\affiliation{The Chinese University of Hong Kong,\\
New Territories, Hong Kong, China}

\emailAdd{yingchan@cuhk.edu.hk}
\emailAdd{tlan@cuhk.edu.hk}
\emailAdd{wulinch@link.cuhk.edu.hk}

\abstract{Given a modular tensor category $\cC$, we construct an associative algebra $\torc$, which we call the torus algebra. We prove that the torus algebra is semisimple by explicitly constructing all the simple modules. Suppose that a topological ordered phase described by $\cC$ is put on a torus. Physically, each simple module over $\torc$ consists of the low energy states on the torus with one anyon excitation, or equivalently, the ground states on a punctured torus where the anyon is enclosed by the puncture. Elements in $\torc$ can be physically interpreted as anyon hopping processes on the torus. We give the precise formula how an arbitrary logical operator on the low energy states on a torus can be realized by moving anyons on the torus. Our work thus provides a theoretical proposal that the low energy states on a torus can serve as topological qudits and one can arbitrarily manipulate them by moving anyons around.}

\begin{document}
\maketitle
\flushbottom

\allowdisplaybreaks
\section{Introduction}
In \cite{Ma} Ma et al.~proposed an operator algebra approach to calculate the ground state degeneracy of the Ising cage-net model~\cite{Prem_2019}. In this paper, we restrict our interest to the operator algebra corresponding to 2+1D topological orders~\cite{Wen89,WN90,Wen90,Wen91,Wen9506066}, which is constructed based on the data of a modular tensor category (MTC) $\cC$ and physically can be thought of as the anyon hopping processes on the torus. 
We refer to such operator algebra as the \emph{torus algebra}, denoted by $\torc$. 

Ref.~\cite{Ma} conjectured a method to decompose the torus algebra, by quotienting out certain ``loop" elements. Unfortunately, this method does not work in general cases.
In this work we give a serious treatment of the torus algebra. The semisimplicity is rigorously proved by constructing all the simple modules. Moreover, we find that the simple modules have clear physical meanings: they are nothing but the low energy subspace on the torus with one anyon excitation. More precisely, given a punctured torus where anyon $q$ is enclosed in the puncture, the corresponding low energy subspace, denoted by $M_q$, forms a simple module over $\torc$. The algebra action can be intuitively depicted by the following topological quantum field theory (TQFT) picture in Figure \ref{ptorus}. This picture agrees with the universal skein theory stated in \cite{bartlett2015modular}: the extended 3d TQFT associated with MTC $\cC$  assigns the punctured torus to the space of ``internal string diagrams’’ inside the torus, which can be identified with the simple modules $M_q$.
\begin{figure}
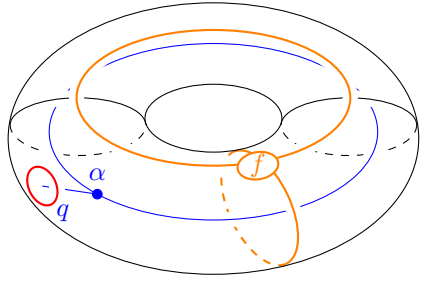

    \label{ptorus}
    \centering
    \ptorus
    \caption{TQFT picture of punctured torus. The orange part denotes the elements of $\torc$. The blue part denotes modules over $\torc$. And the red circle is the puncture.}
    \label{fig:enter-label}
\end{figure}

The main result of this work is the isomorphism $\torc\cong \oplus_q \Endo(M_q)$ (Theorem~\ref{thm.main}). Based on this result, we know that any logical operator in $\Endo(M_q)$ can be realized by elements in $\torc$, i.e., certain patterns of anyon hopping on the torus. Therefore, our work paves the way for exploiting the low energy states on the torus as topological qudits, by manipulating them via moving anyons around the torus.

The paper is organized as follows. In Section \ref{sec2}, we introduce the necessary preliminaries and fix our notations. In Section \ref{sec3}, we give the precise definition of torus algebra. In Section \ref{sec4}, we discuss special elements in the torus algebra, dubbed the centrifugal loops, which lead to a set of central orthogonal idempotents. In Section \ref{sec5}, we construct the simple modules and prove the main theorem. In Section \ref{sec6}, we discuss the extra properties of $\torc$ when $\cC$ is a unitary MTC. In Section \ref{sec7}, we discuss the modular transformation on the punctured torus as an application of our main result. In Section \ref{sec8}, we give some simple but nontrivial examples.






\section{Preliminaries}\label{sec2}
We assume that the reader has necessary knowledge on tensor categories and graphical calculus (see e.g., \cite{etingof2016tensor,Kitaev_2006,Bojko}). In this section we fix our notations and list important results to be used later.

Given a MTC $\cC$, we use $\Irr(\cC)$ to denote the set of (isomorphism classes) of simple objects in $\cC$. For $i,j,k \in \Irr(\cC)$, we depict the basis vectors in $\Hom(i\ot j,k)$ and $\Hom(k,i\ot j)$ by $$\Rmovebb, \quad \Rmoveaa $$ respectively. The morphisms in our graphs are to be composed from bottom to top.

 The dimension of $\Hom(i\ot j,k)$ is denoted by $N_{ij}^k$, called the fusion rules.
The dual object of $i$ is denoted by $\bar i$ and depicted as $$\dirline = \dirlineb.$$ We will add arrow to the lines only when necessary; when there is no arrow it is understood as a line with an upwards arrow.
The quantum dimension of $i$ is denoted by $d_i$, graphically $$d_i= \qdim.$$ We take the normalization such that the evaluation is 
\begin{equation}\label{normalise}
    \compupdown = 
    \sqrt{\frac{d_i d_j}{d_k}}\delta_{\textcolor{purple}{\alpha} \textcolor{brown}{\beta}}\delta_{kk'}\lineq.
\end{equation}

We adopt the convention of F-move as follows.
\begin{equation}
    \Fmovec = \sum\limits_{r,\textcolor{red}{\bullet},\textcolor{blue}{\bullet}}\left[F^{ijk}_{l}\right]_{s\textcolor{orange}{\bullet}\textcolor{green}{\bullet},r\textcolor{red}{\bullet}\textcolor{blue}{\bullet}} \Fmovea.
\end{equation}
The colored nodes will not be labeled by letters unless necessary.
The braiding is depicted by $c_{ij}=\cbraid$ and the convention for $R$-move is 
\begin{equation}
    \Rmoved= \sum_{\textcolor{blue}{\bullet}}R^{ij}_{k\textcolor{blue}{\bullet}\textcolor{red}{\bullet}} \Rmoveaa.
\end{equation}

The $S$-matrix is given by
\begin{equation}
    S_{ab} = \frac{1}{D}\Sfig,
\end{equation}
where $D=\sqrt{\sum_a d_a^2}$ is the total quantum dimension. For a MTC, the $S$ matrix is unitary and one has the Verlinde formula:
\begin{equation}
    N_{ab}^{c} = \sum_x \frac{S_{ax} S_{bx} \overline{S}_{cx}}{S_{1x}},
\end{equation}
 where $\overline S_{cx}$ denotes the complex conjugate.
 
Also, a twist is
\begin{equation}
    \theta_x \linexid = \twist,
\end{equation}
where $\theta_x$ is a phase factor. The multiplicative chiral central charge is defined as
\begin{equation}
    \Theta = \sum_a \frac{d_a^2}{D}\theta_a = e^{2\pi i\frac{c}{8}}, 
\end{equation}
where $c$ is the additive chiral central charge.

The following results will be used extensively in this paper:

\begin{lem}\label{FRmove}
\begin{enumerate}
    \item 
\begin{equation}
    \circline = \frac{S_{ax}}{S_{1x}}\linexid .
\end{equation}
\item One can apply an F-move and the corresponding inverse F-move together on two matching patches of the graph, 
    \begin{equation}\label{Fmove}
\sum\limits_{r,\textcolor{red}{\bullet},\textcolor{blue}{\bullet}}
   \Fmovea \otimes_\mathbb{C} \Fmoveb = \sum\limits_{s,\textcolor{orange}{\bullet},\textcolor{green}{\bullet}}\Fmovec\otimes_\mathbb{C}\Fmoved.
\end{equation}
\item Likewise for R-moves,
    \begin{equation}\label{Rmove}
    \Rmovea\otimes_\mathbb{C} \Rmoveb = \Rmovee = \Rmoved\otimes_\mathbb{C} \Rmovec.
\end{equation}
In the above graph, we follow the nice Einstein-like convention of \cite{henriques2016bicommutant,green2023enriched}, that the pair of nodes with the same color is automatically summed over and the summation symbol is omitted. From now on, we will keep using this convention. 
\item In the Einstein-like convention, we have 
\begin{equation}
    \splitt = \sum_k\sqrt{\frac{d_k}{d_id_j}} \splittb.
\end{equation}
\end{enumerate}

\end{lem}

\section{Torus algebra} \label{sec3}
In this section, we will construct the torus algebra, motivated by stacking anyon hopping operators on a torus.
We cut the torus along two non-contractible loops, and study the resulting graph on a square:
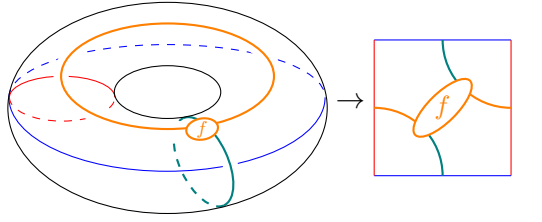
\begin{equation*}
    \mtorus \rightarrow \anytraj\quad.
\end{equation*}
This allows us to intuitively identify the basis vector $f$ as \fsbasis. We then define the stacking of anyon hopping operators as the multiplication,
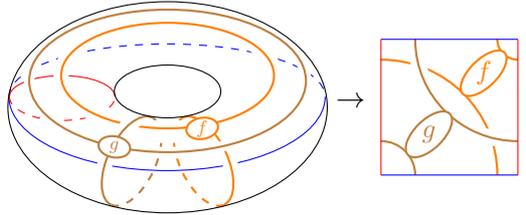
\begin{equation*}
    \mtorusstack \rightarrow \anytrajstack.
\end{equation*}
We now give the rigorous definition:
\begin{definition}
    Given a MTC $\cC$, the torus algebra $\torc$ is the following vector space 
    \[ \torc:= \bigoplus_{i,j\in \Irr{\cC}} \Hom(i\ot j,j\ot i),\]
    equipped with the multiplication map
    \begin{align}
        \star:\torc\otr[\C] \torc &\to \torc \nonumber\\
        g \otr[\C] f &\mapsto g\star f=\sum_{m,n}\sqrt{\frac{d_m d_n}{d_dd_bd_ad_c}}\tordef,
    \end{align}
    and unit $1=\id_\one$.
    \end{definition}

    \begin{remark}
    We give a graphical proof that the multiplication in the above definition is associative, i.e. ,
    \begin{equation}
        \left(f\star g \right)\star h = f\star (g\star h).
    \end{equation}
    \end{remark}
    \begin{proof} 
        This proof makes use of Equation (\ref{Fmove}) from Lemma \ref{FRmove}  twice, in the $\triangle$-equality sign respectively, i.e. ,
        \begin{align}
        (f\star g)\star h &=
        \sum_{m,n,r,s}\sqrt{\frac{d_r d_s}{d_ad_bd_cd_d d_ed_f}}\assoproof{-1}
        \nonumber\\
        &\overset{\triangle}{=} \sum_{m,q,r,s} \sqrt{\frac{d_r d_s}{d_ad_bd_cd_d d_ed_f}}\assoproof{0} \nonumber\\
        &\overset{\triangle}{=} \sum_{p,q,r,s}\sqrt{\frac{d_r d_s}{d_ad_bd_cd_d d_ed_f}}\assoproof{2} = f\star (g\star h).
    \end{align}
    \end{proof}

\begin{remark}
    One should not confuse the torus algebra with the tube algebra (see e.g.,~\cite{ocneanu1994chirality,hardiman2020graphical,LW1311.1784}. They have different background manifolds and the multiplications are given by stacking versus gluing.
\end{remark}

\section{Central idempotent decomposition} \label{sec4}
In this section, we will introduce the central elements in $\torc$ and the construction of central idempotents to decompose the torus algebra systematically. 

Firstly, consider the following graph
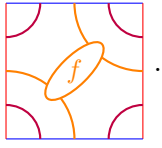
\begin{equation*}
    \centralmo\ .
\end{equation*} 
Since the multiplication is stacking graphs, and using the naturality of braiding, the purple loop can be made not overlap with a general element $f$, we know the result is the same no matter the purple loop is above or below $f$. In other words, the purple loop gives rise to a central element. More rigorously, we introduce 
\begin{definition}[Centrifugal loop]
    We define the \emph{centrifugal loop} elements $\mwasylozenge_a$  in $\torc$, for each $a\in \Irr(\cC)$, to be the elements of the following special form
\begin{equation}
    \lwasylozenge_a \coloneqq \sum_{k,y}\frac{\sqrt{d_k d_y}}{d_a^2}\loopc=\squareloopa \in \torc
\end{equation}
    and represent them graphically as loops surrounding the four corners of the square. Gluing parallel ends of the square makes them contractible loops on the torus; however, such deformation is forbidden in $\torc$ and they are nontrivial elements.
\end{definition}

\begin{prop}\label{prop.cen}
    Centrifugal loops $\mwasylozenge_a$ of any $a\in\Irr(\cC)$ are in the center of $\torc$, i.e. , 
    \begin{equation}
        \lwasylozenge_a \star g = g\star \lwasylozenge_a \mbox{       }\forall g\in \torc.
    \end{equation}
\end{prop}

\begin{proof}
    To show centrifugal loop of any object $a\in\Irr(\cC)$ commutes with any elements  $f\in\torc$ requires replacing braiding-antibraiding pairs with antibraiding-braiding pairs, which is illustrated in Equation (\ref{Rmove}) from Lemma \ref{FRmove}. Implementing Lemma \ref{FRmove} of R-move in $\blacktriangle$-equality sign shows the desired commutativity, i.e. ,
    \begin{align}
    \lwasylozenge_a \star g &= \sum_{k,y,m,n,c,d}\sqrt{\frac{d_md_n}{d_kd_dd_yd_c}}\centerloop{-1} \overset{\triangle}{=} \sum_{r,s,m,n,c,d}\sqrt{\frac{d_md_n}{d_kd_dd_yd_c}}\centerloop{0}\nonumber \\
    &\overset{\blacktriangle}{=} \sum_{r,s,m,n,c,d}\sqrt{\frac{d_md_n}{d_kd_dd_yd_c}}\centerloop{1} \overset{\triangle}{=} \sum_{p,q,m,n,c,d}\sqrt{\frac{d_md_n}{d_kd_dd_yd_c}}\centerloop{2} = g\star  \lwasylozenge_a.
\end{align}
\end{proof}

\begin{prop}
    The muliplication of centrifugal loop elements obey the fusion rule of the MTC, i.e.,
    \begin{equation}
        \lwasylozenge_a \star \lwasylozenge_b =  \sum_c N_{ab}^c\mbox{  }\lwasylozenge_c.
    \end{equation}
\end{prop}
\begin{proof}
    \begin{equation}
    \lwasylozenge_a \star \lwasylozenge_b = \squarestack = \sum_c N_{ab}^c\mbox{  }\squareloopc = \sum_c N_{ab}^c\mbox{  }\lwasylozenge_c
\end{equation}
\end{proof}

By the Verlinde formula, the fusion rule can be diagonalized by the $S$-matrix. Thus, for simple $q\in\Irr(\cC)$, we denote \begin{equation}\label{idempotent}
    P_q^{\wasylozenge} := \sum_a S_{q\mathbf{1}}\overline{S_{aq}}\lwasylozenge_a.
\end{equation} Then,

\begin{prop}
$P_q^{\wasylozenge}$ are mutually orthogonal idempotents, i.e.,
    \begin{equation}
        P_s^{\wasylozenge} \star P_t^{\wasylozenge} = \delta_{st} P_s^{\wasylozenge}.
    \end{equation}
    They are also central and sum to the unit of the algebra.
\end{prop}

\begin{proof}
\begin{align}
    P_s^{\wasylozenge} \star P_t^{\wasylozenge} &= \sum_{ab} S_{s\mathbf{1}}\overline{S_{as}}S_{t\mathbf{1}}\overline{S_{bt}}\lwasylozenge_a \star \lwasylozenge_b = \sum_{abc} S_{s\mathbf{1}}\overline{S_{as}}S_{t\mathbf{1}}\overline{S_{bt}} N_{ab}^{c} \lwasylozenge_c \nonumber \\
    &= \sum_{abcl}S_{s\mathbf{1}}\overline{S_{as}}S_{t\mathbf{1}}\overline{S_{bt}}\frac{S_{al} S_{bl} \overline{S_{cl}}}{S_{\mathbf{1}l}}\lwasylozenge_c = \delta_{st} \sum_c S_{s\mathbf{1}}\overline{S_{cs}}\lwasylozenge_c = \delta_{st} P_s^{\wasylozenge}
\end{align}
That $P_q^{\wasylozenge}$ is central is due to Proposition \ref{prop.cen}. By direct calculation,
\begin{equation}
    \sum_s P_s^{\wasylozenge} = \sum_{sa} S_{s\mathbf{1}}\overline{S_{as}}\lwasylozenge_a = \lwasylozenge_\mathbf{1} = 1
\end{equation}
\end{proof}

\begin{corollary}
$\torc$ admits the direct sum decomposition $\torc=\bigoplus_q P_q^{\wasylozenge} \torc$ as algebras.
\end{corollary}
We will further prove that each direct summand is simple (a matrix algebra) by explicitly constructing the corresponding simple modules.

\section{Simple modules and punctured torus} \label{sec5}

\begin{definition}
    For any simple object $q \in \Irr(\cC)$, we define the left $\torc$-module, denoted by $M_q$,  as the following vector space 
    \begin{equation}
        M_q \coloneqq \bigoplus_{k \in \Irr(\cC)} \Hom\left(k,k\otimes q\right)
    \end{equation}
    equipped with the action of $\torc$,
    \begin{align}
        \rho_q: \torc &\to \Endo\left(M_q\right)\\
        \vartriangleright : \torc\otr[\C] M_q &\to M_q \nonumber\\
        f\otr[\C] \alpha &\mapsto f \vartriangleright \alpha = \sum_{n}\sqrt{\frac{d_n}{d_a d_k}}\action
    \end{align}
    Alternatively, the action can be expressed as the algebra homomorphism $\rho_q:\torc\to \Endo(M_q)$ by $\rho_q(f)\alpha=f\vartriangleright \alpha$.
\end{definition}
\begin{remark} We give a graphical proof that the action in the above definition is associative.
    \begin{equation}
        \left(f\star g\right) \vartriangleright \alpha = f \vartriangleright \left( g \vartriangleright \alpha \right)
    \end{equation}
\end{remark}
    
\begin{proof} 
    The proof uses the F-move convention to allow the direct fusion of $g$ to $\alpha$. 
    \begin{align}
        \left(f\star g\right) \vartriangleright \alpha &= \sum_{n,m,r}\sqrt{\frac{d_m d_r}{d_d d_b d_a d_c d_k}}\assoacta \nonumber \\
        &\overset{\triangle}{=} \sum_{s,m,r}\sqrt{\frac{d_m d_r}{d_d d_b d_a d_c d_k}}\assoactb \nonumber\\
        &= \sum_{s,r}\sqrt{\frac{ d_r}{d_a d_c d_k}}\assoactc =  f\vartriangleright \left(g \vartriangleright \alpha\right)
    \end{align}
\end{proof}

\begin{prop}\label{loopedmodule}
    The action of centrifugal loop $\lwasylozenge_a$ on vector $\alpha\in M_q$ is the same vector with a $a$-loop revolving around $q$:
    \begin{equation}
        \lwasylozenge_a\vartriangleright \alpha = \loopactionf=\frac{S_{aq}}{S_{1q}} \alpha.
    \end{equation}
\end{prop}

\begin{proof} 
    Under the action of the centrifugal loop on vector $\alpha$, summing over $p$ produces separated components of $\loopactiontwo$ and $\loopactionone$. Note that $\loopactiontwo$ is a number times $\id_k$, it commutes with $\loopactionone$. Further summing over $n,s$ gives us the same vector with a $a$-loop revolving around $q$.
    \begin{align}
        \lwasylozenge_a\vartriangleright \alpha &= 
        \sum_{p,s}\frac{\sqrt{d_p d_s}}{d_a^2} \loopaca \vartriangleright \loopacb
    = \sum_{p,s,n} \frac{\sqrt{d_p d_s}}{d_a^2} \sqrt{\frac{d_n}{d_s d_k}} \loopactionb \nonumber\\
        &=\sum_{s,n} \frac{\sqrt{d_s}}{d_a} \sqrt{\frac{d_n}{d_s d_k}}\loopactionc
        =\sum_{s,n} \frac{\sqrt{d_s}}{d_a} \sqrt{\frac{d_n}{d_s d_k}}\loopactiond \nonumber \\
        &=\loopactione
        = \loopactionf
    \end{align}
\end{proof}

\begin{prop}
    The image of the idempotents $P_s^{\wasylozenge}$ under the algebra homomorphism $\rho_q:\torc\to \Endo(M_q)$ is $\rho_q(P_s^{\wasylozenge} )=\delta_{qs}\id_{M_q}$.
\end{prop}
\begin{proof} From Proposition \ref{loopedmodule}, using the unitarity of $S$ in the third equal sign returns a module.
    \begin{equation}
    P_s^{\wasylozenge} \vartriangleright \alpha = \sum_a  S_{s\mathbf{1}}\overline{S_{as}} \loopactionf = \sum_a S_{s\mathbf{1}}\overline{S_{as}} S_{aq} \frac{D}{d_q} \loopacb = \delta_{sq} \loopacb
\end{equation}
\end{proof}

\begin{theorem}\label{thm.sur}
    $\rho_q:\torc\to \Endo({M_q})$ is surjective. The element in $\torc$ that produces a given linear map in $\Endo({M_q})$ can be explicitly constructed as the following. 
\end{theorem}

\begin{proof}
    Pick a basis $\{|k,\gamma\ra:=\kgamma\}\subset M_q$, and the corresponding dual basis $\{\dkgamma\}$ in $\Hom(k\ot q, k)$, such that the normalization follows equation (\ref{normalise}).
    
    For any ${A}_q\in \Endo\left(M_q\right)$, with matrix elements ${A}_q|m,\mu\ra = \sum_{l,\beta}{A}_q^{l\beta,m\mu} |l,\beta\ra$, we can construct $\invrho{A}_q\in \torc$ by \begin{equation}
    \invrho{A}_q = \sum_{a,n, l, \beta,\textcolor{purple}{\gamma}, k} \frac{d_a}{D^2}\sqrt{\frac{d_nd_k}{d_l}}\frac{\sqrt{d_q}}{d_k} {A}_q^{l\beta,k\gamma}\fbasis.
\end{equation}
We now check that $\rho_q(\invrho{A}_q) = {A}_q$:
\begin{align}
    &\invrho{A}_q\vartriangleright |m,\mu\ra = \sum_{a,n,r, l, \beta,\textcolor{purple}{\gamma},k} \frac{d_a}{D^2} \sqrt{\frac{d_r}{d_l}}\frac{\sqrt{d_q}}{d_k}{A}_q^{l\beta,k\gamma} \simpa \nonumber \\
&= \sum_{a,s,r, l, \beta,\textcolor{purple}{\gamma},k} S_{\one\one}\overline{S_{a\one}} \sqrt{\frac{d_r}{d_l}} \frac{\sqrt{d_q}}{d_k}{A}_q^{l\beta,k\gamma}
\simpb\nonumber \\
&= \sum_{s,r, l, \beta,\textcolor{purple}{\gamma}, k} \frac{S_{\one \one}}{S_{\one s}} \delta_{s\one}\delta_{km}\sqrt{\frac{d_r}{d_l}} \frac{\sqrt{d_q}}{d_k}{A}_q^{l\beta,k\gamma} \simpc\nonumber \\
&= \sum_{r, l, \beta, \textcolor{purple}{\gamma}}\delta_{rl}\sqrt{\frac{d_r}{d_l}}\frac{\sqrt{d_q}}{d_m}{A}_q^{l\beta,{m\gamma}}  \simpd =\sum_{l,\beta}{A}_q^{l\beta,m\mu} \modbeta= {A}_q|m,\mu\ra
\end{align}

\end{proof}
\begin{corollary}
    $M_q$ is a simple $\torc$-module.
\end{corollary}

Since $P_s^{\wasylozenge}\vartriangleright M_q=\delta_{qs}M_q$, $P_q^{\wasylozenge}\torc\to \Endo{M_q}$ is also surjective. Taking the direct sum over $q$ on both sides we get an algebra homomorphism
$$\rho=\oplus_q \rho_q:\torc = \bigoplus_q P_q^{\wasylozenge} \torc \to \bigoplus_q \Endo\left(M_q\right).$$
\begin{theorem}\label{thm.main}
$\rho: \torc\to \bigoplus_q \Endo\left(M_q\right)$ is an algebra isomorphism.
\end{theorem}
\begin{proof}
    $\rho:\torc\to \bigoplus_q \Endo\left(M_q\right)$ is surjective. Computing the dimension of $\torc$ and $\bigoplus_q\Endo\left(M_q\right)$, which are both equal to $ \sum_{q\in\Irr(\cC)} \frac1{S^2_{\one q}}$, we know it is bijective.
\end{proof}
\begin{remark}
Note that the element $\invrho A_q$ constructed in the proof of Theorem~\ref{thm.sur} has the property $\invrho A_q\vartriangleright M_s=0$ for $s\neq q$. We know the inverse to $\rho: \torc\to \bigoplus_q \Endo\left(M_q\right)$ is given by $\oplus_q{A}_q\mapsto \sum_q \invrho A_q$.
\end{remark}
\begin{corollary}
\begin{enumerate}
    \item $P_q^{\wasylozenge}\torc\to \Endo(M_q)$ is an algebra isomorphism.
    \item $P_q^{\wasylozenge}$ is a primitive central idempotent.
    \item $\left\{\lwasylozenge_a\right\}$ span the center of $\torc$.
\end{enumerate}
\end{corollary}

\section{Unitary structures: involution, inner product on module, unitary module}\label{sec6}
In physical application, we are usually interested in the case that the MTC $\cC$ is unitary. Denote the unitary structure of $\cC$ by $\dag$. We can further equipped $\torc$ with a positive involution and the modules $M_q$ are moreover unitary. 
\begin{definition}
    $\torc$ has an involution $\ddag: \torc \to \torc$, defined by 
    \begin{equation}
        \involutea=\involuteb.
    \end{equation}
\end{definition}

\begin{remark}
    We give a graphical proof that $\ddag$ is indeed an involution. It is anti-linear by the anti-linearity of $\dagger$. It squares to identity, i.e.,
     \begin{equation}
    f^{\ddag\ddag} =
        \invotwo= \finvo =f.
    \end{equation}
   It is also an algebra antihomomorphism, i.e.,
    \begin{equation}
        \left(g\star f\right)^{\ddag} = f^{\ddag} \star g^{\ddag},
    \end{equation}
    graphically,
    \begin{align}
    \left(g\star f\right)^{\ddag} &= 
    \sum_{m,n}\sqrt{\frac{d_m d_n}{d_dd_bd_ad_c}} \antilina \nonumber\\
    &= \sum_{m,n}\sqrt{\frac{d_m d_n}{d_dd_bd_ad_c}}\antilinb\nonumber\\
    &=\sum_{m,n}\sqrt{\frac{d_m d_n}{d_dd_bd_ad_c}}\antilinc =f^{\ddag} \star g^{\ddag}.
    \end{align}
\end{remark}

    \begin{prop}
        The anitilinear map $\ddag$ is a positive involution map, i.e., 
            \begin{equation}
                f^{\ddag} \star f = 0 \Rightarrow f = 0.
            \end{equation}
    \end{prop}

    \begin{proof}
        \begin{equation}
    f\star f^{\ddag} = \sum_{m,n} \frac{\sqrt{d_md_n}}{d_ad_b} \posinvoa = 0.
    \end{equation}
    In particular, the $m=n=\one$ term in the above sum is 0, i.e.
    \begin{equation}
         \posinvob = 0 \Rightarrow \finvo=0.
    \end{equation}
    \end{proof}
    
\begin{remark}The positive involution also implies that the algebra is semisimple. This is the approach of \cite{Ma} to argue the semisimplicity of the torus algebra. Nonetheless, we have provided a proof for semisimplicity independent from unitary structure and positive involution.
\end{remark}

\begin{definition}
    The inner product on $M_q$ is defined by
    \begin{equation}
    \innera = \frac{\delta_{k,l}}{\sqrt{d_l d_q d_n}}\innerb.
    \end{equation}
\end{definition}

\begin{prop}
    $M_q$ is an unitary module, i.e.,
    
    \begin{equation}
        \la f \vartriangleright \beta | \alpha \ra = \la \beta | f^{\ddag} \vartriangleright \alpha \ra,
    \end{equation}
    equivalently $\rho_q(f^\ddag)=\rho_q(f)^\dag$
\end{prop}

\begin{proof}
     \begin{align}
        \la f \vartriangleright \beta | \alpha \ra &= \sum_{n} \sqrt{\frac{d_n d_1}{d_a d_k^2 d_{k^*}}}\delta_{nk}\unitmoda
        =  \sqrt{\frac{d_n d_1}{d_a d_k^2 d_{k^*}}}\unitmodb\nonumber \\
         &= \sqrt{\frac{d_n d_1}{d_a d_k^2 d_{k^*}}}\unitmodc = \la \beta | f^{\ddag}  \vartriangleright \alpha \ra  \end{align}
\end{proof}
   
\section{Punctured modular transformation}\label{sec7}
We consider the modular transformation on the punctured torus in this section. To begin with, we define a linear map $ S_q\in \Endo(M_q)$ for each simple $q$ via
\begin{equation}
         S_q  \modalpha = \sum_{m} \frac{d_m}{D} \soper.
    \end{equation}
    It has the inverse
    \begin{equation}
         S_q^{-1} \modalpha = \sum_{m} \frac{d_m}{D} \insoper,
    \end{equation}
     which can be checked explicitly
\begin{align}
     S_q^{-1}\left( S_q\left(\alpha\right)\right) &= \sum_{m,n,b} \frac{d_m d_b}{D^2} \unitsqa = \sum_{m,b} \frac{d_m d_b}{D^2} \sqrt{\frac{d_n}{d_k d_m}}\unitsqb \nonumber \\
    &= \sum_{m,n,b} \sqrt{\frac{d_n}{d_k d_m}} \overline{S_{m\one}} \mbox{ }\overline{S_{b\one}} \frac{S_{bn}}{S_{1n}} \unitsqc = \unitsqd = \alpha
\end{align}

To see that such linear map physically corresponds to ``rotating" the (square presentation of) torus by $90^\circ$, we note that on the vector space $M_q$ one can associate a different $\torc$ action, by ``rotating" the elements in $\torc$ by $90^\circ$. We denote such $\torc$-module by $M_q^{\circlearrowleft}$, with the same underlying vector space as $M_q=\oplus_k \Hom(k, k\otimes q)$, and the following action 

\begin{align*}
\blacktriangleright : \torc\otimes_\C M_q &\to M_q,\\    
f\otimes \alpha&\mapsto f\blacktriangleright \alpha=\sum_{n}\sqrt{\frac{d_n}{d_k d_a}}\solidtri
\end{align*}

\begin{theorem}
    $ S_q$ is a $\torc$-module map from $M_q^\circlearrowleft$ to $M_q$, i.e., for any $f\in \torc$ and $\alpha \in M_q$,
    \begin{equation}
         S_q (f\blacktriangleright \alpha)=f\vartriangleright ( S_q \alpha).
    \end{equation}
\end{theorem}
\begin{proof}
    \begin{align}
         &S_q (f\blacktriangleright \alpha) 
         = \sum_{m,n} \frac{d_m}{D}\sqrt{\frac{d_n}{d_k d_b}}\sqactiona 
         = \sum_m \frac{d_m}{D} \sqrt{\frac{d_n}{d_k d_b}}\sqrt{\frac{d_k d_b}{d_n}}\sqactionb \nonumber \\ 
         &= \sum_{m,s} \frac{d_m}{D}\sqrt{\frac{d_s}{d_m d_a}}\sqactionc 
         = \sum_{m,s} \frac{d_m}{D}\sqrt{\frac{d_s}{d_m d_a}}\sqactiond
         = f\vartriangleright ( S_q \alpha).
    \end{align}
\end{proof}

Invoking Theorem~\ref{thm.sur} we can construct an element $\invrho S_q\in\torc$
\begin{equation}
   \invrho S_q = \sum_{k,\textcolor{magenta}{\alpha},m,a,n} \frac{d_a}{D^3}\sqrt{\frac{d_n d_m d_q}{d_k}}\sqelement 
\end{equation}
such that $\rho_s(\invrho S_q)=\delta_{sq}  S_q$, and further an element $S=\sum_q \invrho S_q=\rho^{-1}(\oplus_q S_q)\in\torc$. More precisely,
\begin{equation}
     S=\sum_{a,k,m,n}\frac{d_a}{D^3}\sqrt{d_n d_m d_k}\sqelementb
\end{equation}
such that for any $\alpha\in \oplus_q M_q$, one has $S\vartriangleright \alpha=\oplus_q  S_q(\alpha)$. $S$ is the preimage of $\oplus_q  S_q$ under the isomorphism $\rho:\torc\cong \oplus_q \Endo(M_q)$. 

Since $ S_q$ are invertible, we conclude that $S$ is also invertible in $\torc$ with $S^{-1}=\sum_q \invrho S_q^{-1} =\rho^{-1}(\oplus_q  S_q^{-1})$. The $90^\circ$ rotated action $\blacktriangleright$ of $M_q^\circlearrowleft$ is in fact induced by conjugation by $S$:
\begin{align*}
    f\blacktriangleright \alpha=\oplus_q  S_q^{-1}(f\vartriangleright(\oplus_q  S_q \alpha)=S^{-1}\vartriangleright(f\vartriangleright(S\vartriangleright\alpha))=(S^{-1}\star f\star S)\vartriangleright \alpha.
\end{align*}
Therefore, we can say that $S\in\torc$ realizes the S-transformation on the punctured torus.

The linear map of charge conjugation operator $C_q \in \Endo(M_q)$ for each simple $q$ is defined as
\begin{equation}
    C_q \modalpha = \theta_q\Coper.
\end{equation}

Whereas the T-transformation on the punctured torus is defined via
\begin{equation}
    T_q \modalpha = \Toper.
\end{equation}

The corresponding elements of $T, C\in \torc$ can be found similarly as the construction of $S \in \torc$ above. We now have a complete set of operators $S$, $T$, and $C$.
One can easily show that $S_q^2 = C_q$, $C_q^2 = \theta_q^3 I$, and $(S_qT_q)^3 = \Theta\theta_q^{-2} C_q$. 

\begin{remark}
    In \cite{Kitaev_2006}, Kitaev also introduced punctured $T,S$ transformations. Denoting Kitaev's $T,S$ by $T_q'$ and $S_q'$, they are related to our definition by $T_q'S_q'=T_qS_q$. Kitaev's convention for punctured $T,S$ transformations is also adopted in \cite{wen2019distinguish,Lan_2020}.
\end{remark}

\section{Examples}\label{sec8}
In this section, we demonstrate our systematic approach to identify all the central idempotents. We will decompose the torus algebras in two of the most commonly known anyon models, the chiral Ising anyon model and the Fibonacci anyon model, whose explicit data could be found in e.g.,~\cite{Bonderson_2008}.
We will use the following basis vectors of torus algebras
\begin{equation}\label{basis}
    v\left(a,b,c\right) := \boxbasis\quad.
\end{equation} 
Since the two examples are all multiplicity-free where $N_{ij}^k\leq 1$, there are no extra degrees of freedom on the vertices, and thus no vertex labels.
\subsection{The Chiral Ising Anyon Model}
In \cite{Ma} Ma et al.~decomposed the torus algebra of the chiral Ising anyon model by taking a ``quotient'' of $\torc$. First, they listed all the basis vectors in terms of $v\left(a,b,c\right)$ in $\torc$. Second, by noticing the dimension of $\torc$ is $10$, they recognized that $\torc$ consists of one block with dimension $3$ denoted as $\mathrm{Mat}_3$ and another with dimension $1$ denoted as $\mathrm{Mat}_1$. Lastly, they find the projector of $\mathrm{Mat}_3$ such that $P\torc=\mathrm{Mat}_3$. They used the intuition that the centrifugal loop of a non-abelian anyon (which is a non-trivial element in $\torc$) should be set to its quantum dimension, and by taking a quotient forcing such conditions, they obtain the projector corresponding to the ground state subspace on the torus. Based on our general discussion on $\torc$, indeed by setting $\lwasylozenge_a=d_a=DS_{a\one}$ one has
  \begin{equation}
     P_q^{\wasylozenge} = \sum_a S_{q\mathbf{1}}\overline{S_{aq}}d_a=\delta_{q\one}.
 \end{equation}
 Thus, this method is good enough to obtain the subspace $P_\one^{\wasylozenge}\torc$ where $P_\one^{\wasylozenge}=1$ and other $P_q^{\wasylozenge}=0$. However, it fails to produce the entire decomposition of $\torc$ unless there are only two blocks. 

In our approach, we first obtain centrifugal loops of all anyon charges, then we can calculate all the idempotents using equation (\ref{idempotent}) to decompose $\torc$ completely.

According to \cite{Ma}, we have the following centrifugal loops of anyon charges $\{\one, \sigma, \psi\}$,
\begin{equation}
    \lwasylozenge_\psi = \lwasylozenge_\one = 1 ,\quad \frac{1}{\sqrt{2}}\lwasylozenge_\sigma = : r,
\end{equation}
where $r = \frac{1}{2}\left(1+\psi_x+\psi_y -\psi_x\psi_y\right)$ with $\psi_x = v\left(\psi,\one,\psi\right)$ and $\psi_y = v\left(\one,\psi,\psi\right)$.

With equation (\ref{idempotent}), the idempotents are
\begin{align}
    P_\one^{\wasylozenge} &= \sum_i \frac{1}{D^2}d_i \lwasylozenge_i= \frac{1}{4}\left(\lwasylozenge_\one + \lwasylozenge_\psi + \sqrt{2} \lwasylozenge_\sigma\right) = \frac{1}{2}\left(1+r\right),\\
    P_\psi^{\wasylozenge} &= \frac{1}{4}\left(\lwasylozenge_\one + \lwasylozenge_\psi - \sqrt{2} \lwasylozenge_\sigma\right)= \frac{1}{2}\left(1-r\right),\\
    P_\sigma^{\wasylozenge} &= 0.
\end{align}
We recovered the projector $P_\one^{\wasylozenge}$ from \cite{Ma} in a systematical way, which corresponds to the ground state subspace. It is easy to see restricting to the subspace of $\torc$ with $P_\one^{\wasylozenge}=1$ and $P_\psi^{\wasylozenge}=0$ is equivalent to setting setting $r=1$, i.e., forcing the centrifugal loop $\lwasylozenge_\sigma $ to be its quantum dimension $d_\sigma=\sqrt2$.

\subsection{Fibonacci anyon model}
The Fibonacci model has anyon charges $\{\one, \varepsilon\}$. First, we evaluate the centrifugal loops of all anyon charges in this model, and it is obvious that $\lwasylozenge_\one = 1$. The quantum dimension of $\varepsilon$ is $d_\varepsilon = \frac{1+\sqrt{5}}{2}=: \phi$. The nontrivial entries of the $F$-symbol are \begin{equation*}
    \left(F^{\varepsilon\varepsilon\varepsilon}_\varepsilon\right)_{e,f} = \begin{pmatrix}
        \phi^{-1} & \phi^{-1/2}\\
        \phi^{-1/2} & -\phi^{-1}
    \end{pmatrix}_{e,f}.
\end{equation*}

The evaluation of centrifugal loop $\lwasylozenge_\varepsilon$ is
\begin{align}
    \squareloope &=  \sum_{k,y}\frac{\sqrt{d_k d_y}}{d_\varepsilon^2}\loope = \sum_{k,y,r, s}\frac{\sqrt{d_k d_y}}{d_\varepsilon^2} F^{\varepsilon \varepsilon \bar{y}}_{k, \varepsilon, r} F^{\bar{y} \varepsilon \varepsilon}_{k,\varepsilon, s} \centritrans \nonumber\\
    &= \sum_{r,k,y} \sqrt{\frac{d_k d_y}{d_{\varepsilon}^2 d_r}}F^{\varepsilon \varepsilon \bar{y}}_{k, \varepsilon, r} F^{\bar{y} \varepsilon \varepsilon}_{k,\varepsilon, s} \centritransb = \sum_{r,k,y} \sqrt{\frac{d_k d_y}{d_{\varepsilon}^2 d_r}}F^{\varepsilon \varepsilon \bar{y}}_{k, \varepsilon, r} F^{\bar{y} \varepsilon \varepsilon}_{k,\varepsilon, s} v\left(y, k, r\right)\nonumber \\
    &= \frac{1}{\phi}v\left(\one,\one,\one\right) + \frac{1}{\phi}v\left(\varepsilon,\varepsilon,\one\right) + \frac{1}{\phi}v\left(\varepsilon,\one,\varepsilon\right) +\frac{1}{\phi}v\left(\one,\varepsilon,\varepsilon\right) + \frac{1}{\phi^{5/2}}v\left(\varepsilon,\varepsilon,\varepsilon\right).
\end{align}
Similar to the previous example, we define $$\frac{1}{\phi}\lwasylozenge_\varepsilon =: s.$$

The idempotents are
\begin{align}
    P_\one^{\wasylozenge} &= \sum_{i}\frac{d_i}{D^2}\lwasylozenge_i = \frac{1}{D^2}\left(1+\phi^2\cdot s\right), \\
    P_\varepsilon^{\wasylozenge} &= \frac{d_\varepsilon}{D} \sum_i \overline{S}_{i\varepsilon}\lwasylozenge_i =\frac{\phi}{D}\left(\frac{\phi}{D} -\frac{\phi\cdot s}{D}\right)= \frac{\phi^2}{D^2}\left(1-s\right).
\end{align}
Again, $P_\one^{\wasylozenge}$ corresponds to the ground state subspace, and the subspace of $\torc$ with $P_\one^{\wasylozenge}=1$ and $P_\varepsilon^{\wasylozenge}=0$ is obtained by setting setting $s=1$, i.e., forcing the centrifugal loop $\lwasylozenge_\varepsilon $ to be its quantum dimension $d_\varepsilon=\phi$.


\appendix

\acknowledgments

We would like to thank Holiverse Yang for the discussions. TL is supported by start-up funding from The Chinese University of Hong Kong, and by funding from Research Grants Council, University Grants Committee of Hong Kong (ECS No. 24304722).



\bibliographystyle{JHEP}
\bibliography{biblio.bib}

\providecommand{\href}[2]{#2}\begingroup\raggedright\begin{thebibliography}{10}

\bibitem{Ma}
X.~Ma, A.~Malladi, Z.~Wang, Z.~Wang and X.~Chen, \emph{Ground state degeneracy of the ising cage-net model}, \href{https://doi.org/10.1103/PhysRevB.107.085123}{\emph{Phys. Rev. B} {\bfseries 107} (2023) 085123}.

\bibitem{Prem_2019}
A.~Prem, S.-J.~Huang, H.~Song and M.~Hermele, \emph{Cage-net fracton models}, \href{https://doi.org/10.1103/physrevx.9.021010}{\emph{Physical Review X} {\bfseries 9} (2019) }.

\bibitem{Wen89}
X.G.~Wen, \emph{Vacuum degeneracy of chiral spin states in compactified space}, \href{https://doi.org/10.1103/PhysRevB.40.7387}{\emph{Physical Review B} {\bfseries 40} (1989) 7387}.

\bibitem{WN90}
X.G.~Wen and Q.~Niu, \emph{Ground-state degeneracy of the fractional quantum {H}all states in the presence of a random potential and on high-genus {R}iemann surfaces}, \href{https://doi.org/10.1103/PhysRevB.41.9377}{\emph{Physical Review B} {\bfseries 41} (1990) 9377}.

\bibitem{Wen90}
X.G.~Wen, \emph{Topological orders in rigid states}, \href{https://doi.org/10.1142/S0217979290000139}{\emph{International Journal of Modern Physics B} {\bfseries 04} (1990) 239}.

\bibitem{Wen91}
X.G.~Wen, \emph{Non-{A}belian statistics in the fractional quantum hall states}, \href{https://doi.org/10.1103/PhysRevLett.66.802}{\emph{Physical Review Letters} {\bfseries 66} (1991) 802}.

\bibitem{Wen9506066}
X.-G.~Wen, \emph{Topological orders and edge excitations in {FQH} states}, \href{https://doi.org/10.1080/00018739500101566}{\emph{Advances in Physics} {\bfseries 44} (1995) 405} [\href{https://arxiv.org/abs/cond-mat/9506066}{{\ttfamily cond-mat/9506066}}].

\bibitem{bartlett2015modular}
B.~Bartlett, C.L.~Douglas, C.J.~Schommer-Pries and J.~Vicary, \emph{Modular categories as representations of the 3-dimensional bordism 2-category},  \href{https://arxiv.org/abs/1509.06811}{{\ttfamily 1509.06811}}.

\bibitem{etingof2016tensor}
P.~Etingof, S.~Gelaki, D.~Nikshych and V.~Ostrik, \emph{Tensor Categories}, Mathematical Surveys and Monographs, American Mathematical Society (2016).

\bibitem{Kitaev_2006}
A.~Kitaev, \emph{Anyons in an exactly solved model and beyond}, \href{https://doi.org/10.1016/j.aop.2005.10.005}{\emph{Annals of Physics} {\bfseries 321} (2006) 2}.

\bibitem{Bojko}
B.~Bakalov and A.~Kirillov, \emph{Lectures on tensor categories and modular functors}, {\emph{Amer. Math. Soc. Univ. Lect. Ser.} {\bfseries 21} (2001) }.

\bibitem{henriques2016bicommutant}
A.~Henriques and D.~Penneys, \emph{Bicommutant categories from fusion categories},  \href{https://arxiv.org/abs/1511.05226}{{\ttfamily 1511.05226}}.

\bibitem{green2023enriched}
D.~Green, P.~Huston, K.~Kawagoe, D.~Penneys, A.~Poudel and S.~Sanford, \emph{Enriched string-net models and their excitations},  \href{https://arxiv.org/abs/2305.14068}{{\ttfamily 2305.14068}}.

\bibitem{ocneanu1994chirality}
A.~Ocneanu, \emph{Chirality for operator algebras}, {\emph{Subfactors (Kyuzeso, 1993)} (1994) 39}.

\bibitem{hardiman2020graphical}
L.~Hardiman, \emph{A graphical approach to the drinfeld centre},  \href{https://arxiv.org/abs/1911.07271}{{\ttfamily 1911.07271}}.

\bibitem{LW1311.1784}
T.~Lan and X.-G.~Wen, \emph{Topological quasiparticles and the holographic bulk-edge relation in (2+1)-dimensional string-net models}, \href{https://doi.org/10.1103/PhysRevB.90.115119}{\emph{Physical Review B} {\bfseries 90} (2014) 115119} [\href{https://arxiv.org/abs/1311.1784}{{\ttfamily 1311.1784}}].

\bibitem{wen2019distinguish}
X.~Wen and X.-G.~Wen, \emph{Distinguish modular categories and 2+1d topological orders beyond modular data: Mapping class group of higher genus manifold},  \href{https://arxiv.org/abs/1908.10381}{{\ttfamily 1908.10381}}.

\bibitem{Lan_2020}
T.~Lan, X.~Wen, L.~Kong and X.-G.~Wen, \emph{Gapped domain walls between 2+1d topologically ordered states}, \href{https://doi.org/10.1103/physrevresearch.2.023331}{\emph{Physical Review Research} {\bfseries 2} (2020) }.

\bibitem{Bonderson_2008}
P.~Bonderson, K.~Shtengel and J.~Slingerland, \emph{Interferometry of non-abelian anyons}, \href{https://doi.org/10.1016/j.aop.2008.01.012}{\emph{Annals of Physics} {\bfseries 323} (2008) 2709–2755}.

\end{thebibliography}\endgroup






\end{document}